\documentclass[journal]{IEEEtran}

\usepackage[dvips]{color}
\usepackage{comment}
\usepackage{todonotes}
\usepackage{epsf}
\usepackage{epsfig}
\usepackage{times}
\usepackage{epsfig}
\usepackage{graphicx}
\usepackage{bbold}
\usepackage{mathtools}
\usepackage{mathrsfs}
\usepackage{amssymb}
\usepackage{pdfpages}
\usepackage{epstopdf}
\usepackage{dsfont}
\usepackage{lettrine} 
\usepackage{amsmath,epsfig,amssymb,algorithm,algpseudocode,amsthm,cite,url}
\usepackage{amsmath}
\usepackage{subcaption}
\allowdisplaybreaks
\usepackage{csquotes}
\usepackage{verbatim}
\usepackage[english]{babel}
\usepackage{amsmath,amssymb}

\usepackage[justification=centering]{caption}
\usepackage{verbatim} 
\newtheorem{theorem}{\bf Theorem}

\newtheorem{proposition}{\bf Proposition}

\usepackage{setspace}	
\begin{document}
	
	\title{\vspace{-0mm}Efficient Deployment of Multiple Unmanned Aerial Vehicles for Optimal Wireless Coverage}

	\author{\IEEEauthorblockN{Mohammad Mozaffari$^1$, Walid Saad$^1$, Mehdi Bennis$^2$, and M\'erouane Debbah$^3$}\vspace{0.1cm}\\
		\IEEEauthorblockA{
			\small $^1$ Wireless@VT, Electrical and Computer Engineering Department, Virginia Tech, VA, USA, Emails: \url{{mmozaff , walids}@vt.edu} \\
			$^2$ CWC - Centre for Wireless Communications, Oulu, Finland, Email: \url{bennis@ee.oulu.fi}\\
			$^3$ Mathematical and Algorithmic Sciences Lab, Huawei France R \& D, Paris, France, and CentraleSup´elec,\\   Universit´e Paris-Saclay, Gif-sur-Yvette, France, Email: \url{merouane.debbah@huawei.com}
		}\vspace{-0.9cm}}
\maketitle
\begin{abstract}
	
 In this paper, the efficient deployment of multiple unmanned aerial vehicles (UAVs) with directional antennas acting as wireless base stations that provide coverage for ground users is analyzed. First, the downlink coverage probability for UAVs as a function of the altitude and the antenna gain is derived. Next, using circle packing theory, the three-dimensional locations of the UAVs is determined in a way that the total coverage area is maximized while maximizing the coverage lifetime of the UAVs. Our results show that, in order to mitigate interference, the altitude of the UAVs must be properly adjusted based on the beamwidth of the directional antenna as well as coverage requirements. Furthermore, the minimum number of UAVs needed to guarantee a target coverage probability for a given geographical area is determined. Numerical results evaluate the various tradeoffs involved in various UAV deployment scenarios. 
\end{abstract} \vspace{-0.8cm}

\section{Introduction}\vspace{-0.1cm}
The use of wireless aerial platforms is a promising approach to improve the performance of wireless communication networks. In particular, unmanned aerial vehicles (UAVs) can act as flying base stations to enhance the coverage and rate performance of wireless networks in different scenarios such as temporary hotspots and emergency situations \cite{Mozaffari}. In addition, mobile UAVs can establish efficient communications with ground sensors for alarm messages delivery scenarios \cite{de2013}.  
 Indeed, using UAVs as aerial base stations provides several advantages compared to the terrestrial base stations. First, due to their higher altitude, aerial base stations have a higher chance of line-of-sight (LoS) links to ground users. Second, UAVs can easily move and have a flexible deployment, and hence, they can provide rapid, on-demand communications. Finally, using directional antennas, one can further enhance the UAV-based communications due to the possibility of using effective beamforming schemes \cite{bekmezci}.

Despite the numerous advantages for using UAVs as flying base stations, one must overcome a number of technical challenges. These challenges include the optimal 3D deployment of UAVs, energy limitations, interference management, and path planning \cite{Mozaffari,de2013 ,bekmezci,mozaffari2, orfanus}. In particular, the deployment problem is of paramount importance as it highly impacts the energy consumption as well as the interference generated by UAVs. However, only a limited number of existing literature have addressed the interplay between UAV deployment and wireless performance \cite{ Mozaffari, mozaffari2, orfanus, kosmerl}.  
 For instance, in \cite{orfanus}, the use of multiple UAVs as wireless relays in order to provide service for ground sensors is investigated. This work addressed the tradeoff between connectivity among the UAVs and maximizing the area covered by the UAVs. However, the work in \cite{orfanus} does not consider the use of UAVs as aerial base stations and their mutual interference in downlink communications. In  \cite{kosmerl}, the authors used evolutionary algorithms in order to find the optimal placement of low altitude platforms (LAPs) and portable base stations for disaster relief scenarios. However, the model of \cite{kosmerl} assumes that overlapping LAPs' coverage areas  is allowed by using inter-cell interference coordination (ICIC). However, ICIC requires further communications between LAPs. 

The main contribution of this paper is to investigate the optimal 3D deployment of multiple UAVs in order to maximize the downlink coverage performance while using a minimum transmit power. Given a target geographical area, the coverage requirements of the ground users, and a number of UAVs that use directional antennas, we develop a novel framework to determine the optimal 3D locations of the UAVs. First, we derive the downlink coverage probability for a UAV                                                                                                                                                                                                                                                                                                                                                                                                                                                                                                                                                                                                                                                                                                                                                                                                                                                                                                                                                                                                                                                                                                                                                                                                                                                                                                                                                                                                                                                                                                                                                                                                                                                                                                                                                                                                                                                                                                                                                                                                                                                                                                                                                                                                                                                                                                                                                                                                                                                                                                                                                                                                                                                                                                                                                                                                                                                                                                                                                                                                                                                                                                                                                                                                                                                                                                                                                                                                                                                                                                                                                                                                                                                                                                                                                                                                                                                                                                                                                                                                                                                                                                                                                                                                                                                                                                                                                                                                                                                                                                                                                                                                                                                                                                                                                                                                                                                                                                                                                                                                                                                                                                                                                                                                                                                                                                                                                                                                                                                                                                                                                                                                                                                                                                                                                                                                                                                                                                                                                                                                                                                                                                                                                                                                                                                                                                                                                                                                                                                                                                                                                                                                                                                                                                                                                                                                                                                                                                                                                                                                                                                                                                                                                                                                                                                          as a function of the UAV's altitude and the antenna gain. Next, using  circle packing theory \cite{gaspar}, we propose an efficient deployment method which leads to the maximum coverage performance while ensuring that the coverage areas of UAVs do not overlap. Our results show that, considering the size of the desired area, the number of
available UAVs and the gain (or beamwidth) of the directional antennas, the altitude and locations of the UAVs can be appropriately  adjusted for satisfying  the coverage requirements. In addition, our results reveal the minimum number of UAVs required to guarantee a target coverage for a given geographical area.         

The rest of this paper is organized as follows. In Section II, we present the system model and describe the air-to-ground channel model. Section III presents the coverage analysis and the proposed deployment method. In Section IV, we provide the simulation results, and Section V concludes the paper.\vspace{-0.25cm}

\section{System Model}

Consider a circular geographical area of radius $R_c$, as illustrated in Figure \ref{Model}, within which $M$ UAVs must be deployed to provide wireless coverage for ground users located within the area. In this model, we consider a stationary low altitude aerial platform such as quadrotor UAVs. The UAVs are assumed to be symmetric having the same transmit power and altitude. 
\textcolor{black}{We denote the UAV's directional antenna half beamwidth by $\theta_B$, and, thus, the antenna gain can be approximated by \cite{ven}:\vspace{-0.2cm}
\begin{equation}
	G = \left\{ \begin{array}{l}
		{G_{\text{3dB},}}\,\,\,\,\,\,\,\frac{{ - {\theta _B}}}{2} \le \varphi  \le \frac{{{\theta _B}}}{2},\\
		g(\varphi),\,\,\,\,\,\,\,\,\,\,\,\,\,\,{\text{ otherwise,}}
	\end{array} \right.\vspace{-0.1cm}
\end{equation}
where $\varphi$ is the sector angle, $G_{\text{3dB}} \approx \frac{{29000}}{{\theta _B^2}}$ with $\theta_B$ in degrees, is the main lobe gain \cite{balanis2016}. Also, $g(\varphi)$ is the antenna gain outside of the main lobe.
}
For the air-to-ground channel modeling, a common approach is to consider the LoS and non-line-of-sight (NLoS) links between the UAV and the ground users separately\cite{HouraniModeling}. 
 Each link has a specific probability of occurrence which depends on the elevation angle, environment, and relative location of the UAV and the users. Clearly, for NLoS links the shadowing and blockage loss is higher than the LoS links. Therefore, the received signal power from UAV $j$ at a user's location can be given by \cite{HouraniModeling}:\vspace{-0.15cm}
\begin{equation}\label{Pr}
{P_{r,j}}(dB) = \left\{\hspace{-0.16cm} \begin{array}{l}
{P_t} + {G_{\text{3dB}}} - {L_{dB}} - {\psi _\text{LoS},} \,\,\,\hspace{0.21cm}{\text{LoS link,}}\\
{P_t} + {G_{\text{3dB}}} - {L_{dB}} - {\psi _\text{NLoS},}\hspace{0.21cm}{\text{NLoS link,}}
\end{array} \right.
\end{equation}
where $P_{r,j}$ is the received signal power, $P_t$ is the UAV's transmit power, and $G_{\text{3dB}}$ is the UAV antenna gain in dB. Also, $L_{dB}$ is the path loss which for the air-to-ground communication is:\vspace{-0.1cm}
\begin{equation} 
{L_{dB}} = 10n\text{log} \left(\frac{{4\pi {f_c}d_j}}{c}\right),\vspace{-0.1cm} 
\end{equation}
where $f_c$ is the carrier frequency, $c$ is the speed of light, $d_j$ is the distance between UAV $j$ and a ground user, and $n\ge 2$ is the path loss exponent. Also, ${\psi _\text{LoS}}\sim N({\mu _\text{LoS}},\sigma _\text{LoS}^2)$ and ${\psi _\text{NLoS}}\sim N({\mu _\text{NLoS}},\sigma _\text{NLoS}^2)$ are shadow fading with normal distribution in dB scale for LoS and NLoS links. The mean and variance of the shadow fading for LoS and NLoS links are $({\mu _\text{LoS}},\sigma _\text{LoS}^2)$, and $({\mu _\text{NLoS}},\sigma _\text{NLoS}^2)$. As shown in \cite{HouraniModeling}, the variance depends on the elevation angle and type of the environment as follows:\vspace{-0.2cm}
\begin{align}
{\sigma _\text{LoS}}(\theta_j ) = {k_1}\exp ( - {k_2}\theta_j ),\\
{\sigma _\text{NLoS}}(\theta_j ) = {g_1}\exp ( - {g_2}\theta_j ),\vspace{-0.3cm}
\end{align}
where $\theta_j={\sin ^{- 1}}(h/d_j)$ is the elevation angle between the UAV and the user, $k_1$, $k_2$, $g_1$, and $g_2$ are constant values which depend on environment.
Finally the LoS probability is given by \cite{HouraniModeling}:\vspace{-0.3cm} 
\begin{equation}
{P_{\text{LoS},j}} = \alpha {\left( {\frac{{180}}{\pi }\theta_j  - 15} \right)^\gamma },\vspace{-0.1cm}
\end{equation}
where $\alpha$ and $\gamma$ are constant values reflecting the environment impact. Note that, the NLoS probability is $P_{\text{NLoS},j}=1-P_{\text{LoS},j}$. Next, using this air-to-ground channel model, we derive the downlink coverage probability and characterize an efficient scheme for deploying of multiple UAVs.\vspace{-0.3cm}  

\begin{figure}[!t]
	\begin{center}
		\vspace{-0.1cm}
		\includegraphics[width=6.8cm] {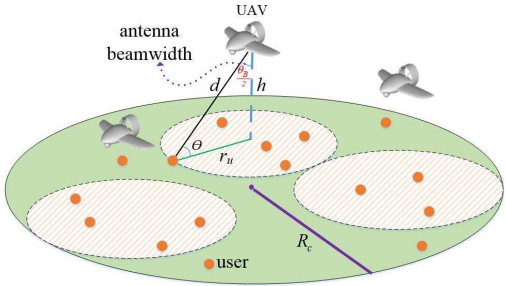}
		\vspace{-0.2cm}
		\caption{ \small System model.}
		\label{Model}
	\end{center}\vspace{-0.8cm}
\end{figure}

\section{Optimal Multi-UAV deployment}\vspace{-0.03cm}
\textcolor{black}{First, we find the coverage radius of each UAV in the presence of interference from other UAVs. To this end, the coverage probability of a single UAV needs to be derived. Then, we propose an efficient deployment strategy for $M$ UAVs that maximizes the total coverage performance while maximizing the coverage lifetime.}\vspace{-0.1cm}
\textcolor{black}{
\begin{theorem}\label{Pcov}
	\normalfont
	The coverage probability for a ground user, located at a distance $r\le h. \rm tan(\theta_B/2)$ from the projection of the UAV $j$ on the desired area, is given by:\vspace{-0.1cm}
	\begin{small}
	\begin{align}
	{P_{{\mathop{\rm cov}} }} &= {P_{\text {\rm LoS},j}}Q\left( {\frac{{{P_\textrm{min}} + {L_{dB}} - {P_t} - {G_{\text{3dB}}} + {\mu _\text{\rm LoS}}}}{{{\sigma _\text{\rm LoS}}}}} \right)\nonumber \\
	 &+ P_{\text {\rm NLoS},j}Q\left( {\frac{{{P_\textrm{min}} + {L_{dB}} - {P_t} - {G_{\text{3dB}}} + {\mu _\text{\rm NLoS}}}}{{{\sigma _\text{\rm NLoS}}}}} \right),
	\end{align}
	\end{small}where ${P_{\min }} = 10\log \left( {\beta N + \beta \bar I} \right)$ is the minimum received power requirement (in dB) for a successful detection, $N$ is the noise power, $\beta$ is the signal-to-interference-plus-noise-ratio (SINR) threshold. $\bar I$ is the mean interference power received from the nearest UAV $k$ which is given by:\\
	\begin{small} $\bar I \approx P_tg({\varphi _k})\left[ {{{10}^{\frac{{ - {\mu _{\text{LoS},k}}}}{{10}}}}{P_{\text{LoS},k}} + {{10}^{\frac{{ - {\mu _{\text{NLoS},k}}}}{{10}}}}{P_{\text{NLoS},k}}} \right]{\left( {\frac{{4\pi {f_c}{d_k}}}{c}} \right)^{ - n}}$.
		\end{small}
	  Also, $Q(.)$ is the $Q$ function.
\end{theorem}\vspace{-0.1cm}}\vspace{-0.7cm} 
\textcolor{black}{\begin{proof}
	 The downlink coverage probability for a ground user considering the mean value of interference between UAVs can be written as: \begin{small}		
	\begin{align} \label{Pcov2}
{P_{{\mathop{\rm cov}} }}&= \mathds{P}\left[\frac{{{P_{r,j}}}}{{N + \bar I}} \ge \beta \right] = \mathds{P}\left[ {{P_{r,j}}(dB) \ge {P_\textrm{min}}} \right] \nonumber\\
& = {P_{{\rm LoS},j}}\mathds{P}\left[ {{{P_{r,j}}(\text{LoS})} \ge {P_{\min }}} \right] + {P_{{\rm{NLoS}},j}}\mathds{P}\left[ {{{P_{r,j}}(\text{NLoS})} \ge {P_{\min }}} \right]\nonumber\\
&\mathop  = \limits^{(a)} {P_{{\rm LoS},j}}\mathds{P}\left[ {{\psi _{\rm LoS}} \le {P_t} + {G_{\text{3dB}}} - {P_\textrm{min}} - {L_{dB}}} \right] \nonumber \\
&+ {P_{{\rm NLoS},j}}\mathds{P}\left[ {{\psi _{\rm NLoS}} \le {P_t} + {G_{\text{3dB}}} - {P_{\textrm{min}}} - {L_{dB}}} \right]\nonumber\\
&\mathop  = \limits^{(b)} {P_{{\rm LoS},j}}Q\left( {\frac{{{P_\textrm{min}} + {L_{dB}} - {P_t} - {G_{\text{3dB}}} + {\mu _{\rm LoS}}}}{{{\sigma _{\rm LoS}}}}} \right)\nonumber\\
& + {P_{{\rm NLoS},j}}Q\left( {\frac{{{P_\textrm{min}} + {L_{dB}} - {P_t} - {G_{\text{3dB}}} + {\mu _{\rm NLoS}}}}{{{\sigma _{\rm NLoS}}}}} \right),
\end{align}	
\end{small}where $\mathds{P}[.]$ is the probability notation, and ${P_{\min }} = 10\log \left( {\beta N + \beta \bar I} \right)$. Clearly, due to the use of directional antennas, interference received from the nearest UAV $k$ is dominant. Hence, $\bar I$ can be written as:
 \begin{align}
 &\bar I \approx {P_{\text{LoS},k}}\mathds{E}\left[ {{P_{r,k}}(\text{LoS})} \right] + {P_{\text{NLoS},k}}\mathds{E}\left[ {{P_{r,k}}(\text{NLoS})} \right] \nonumber\\
 &=P_t g({\varphi _k})\left[ {{{10}^{\frac{{ - {\mu _\textrm{LoS}}}}{{10}}}}{P_{\text{LoS},k}} + {{10}^{\frac{{ - {\mu _{\textrm{NLoS}}}}}{{10}}}}{P_{\text{NLoS},k}}} \right]{\left( {\frac{{4\pi {f_c}{d_k}}}{c}} \right)^{-n}}.\nonumber
 \end{align}\vspace{-0.3cm}\\
   where $\mathds{E}[.]$ is the expected value over the received interference power. The mean interference is a reasonable approximation for the interference and leads to a tractable coverage probability expression. In (\ref{Pcov2}), $(a)$ is a direct result of (\ref{Pr}), and $(b)$ comes from the complementary cumulative distribution distribution function (CCDF) of a Gaussian random variable. 
Furthermore, $r\le h. \rm tan(\theta_B/2)$ implies that a user can be covered only if it is in the coverage range of a directional antenna with beamwidth $\theta_B$. This proves the theorem.
\end{proof} }\vspace{-0.2cm} 

\textcolor{black}{Note that, Theorem 1 provides the coverage probability for users located at any arbitrary range $r$.} From Theorem 1, it is observed that changing the UAV's altitude impacts the coverage by affecting several parameters including the distance between the UAV and users, LoS probability, and the feasible coverage radius ($r\le h. \rm tan(\frac{{{\theta_B}}}{2})$). For instance, increasing a UAV's altitude leads to a higher path loss and LoS probability, as well as a higher feasible coverage radius. \textcolor{black}{ Clearly, in the presence of interference, the UAVs need to increase their transmit power in order to meet the coverage requirements. Furthermore, as the number of UAVs increases, the distance between the UAVs decreases, and hence, the interference from the nearest UAV will increase.} 
The coverage radius of a UAV, $r_u$, is the maximum range within which the probability that users are covered by the UAV is greater than a specified threshold ($\epsilon$). Clearly, $r_u$ depends on the transmit power, antenna beamwidth, $\epsilon$, number of UAVs, and UAVs' locations. 
Thus, the coverage radius is given by:\vspace{-0.05cm}
\begin{equation}\label{ru}
{r_u} = \max \{ r|{P_{{\mathop{\rm cov}} }}(r,{P_t},{\theta _B}) \ge \varepsilon \}.\vspace{-0.2cm}
\end{equation}

Now, consider the geographical area of interest which should be covered by multiple UAVs. The UAVs must be placed in a way to maximize the total coverage, and to avoid any overlapping in their coverage areas. Furthermore, while maximizing the total coverage, each UAV must use a minimum transmit power in order to maximize the coverage lifetime. The coverage lifetime is defined as the maximum time that the UAVs can provide coverage for the given area. The coverage lifetime is maximized when the UAVs have the same transmit power (equal coverage radius). Therefore, assuming a circular coverage area for each UAV, the problem can be formulated as follows:\vspace{-0.25cm}
\begin{align}\label{opt}
&\mathop {(\vec r_j^*,{h^*},r_u^* )}\limits_{i \in \{ 1,...,M\} }  = \arg \max M.r_u^2,\\
{\rm{st.}}\,\,\,&{\rm{||}}{{\vec r}_j}{\rm{ - }}{{\vec r}_k}{\rm{||}} \ge 2{r_u},\,\,\, j \ne k \in \{ 1,...,M\},\label{over} \\
&{\rm{||}}{{\vec r}_j+r_u}{\rm{||}} \le {R_c},\label{Rc}\\
& r_u\le h. \rm tan(\theta_B/2),\vspace{-0.1cm}
\end{align}
where $M$ is the number of UAVs, $R_c$ is the radius of the desired geographical area, $\vec r_j$ is the vector location of UAV $j$ within the 2D plane of the desired area considering the center of the area as the origin, and $r_u$ is the maximum coverage radius of each UAV. Considering the optimization problem in (\ref{opt}), constraint (\ref{over}) ensures that no coverage overlap occurs, and (\ref{Rc}) guarantees that UAVs do not cover outside of the desired area. As a result, the potential interference between UAVs will be avoided, and also, users located outside the given area (undesired area) will not be affected by UAVs' transmissions.
   
Solving (\ref{opt}) is challenging due to the high number of unknowns and the nonlinear constraints. We model this problem by exploiting the so-called \emph{circle packing problem}  \cite{gaspar}. In the circle packing problem, $M$ circles should be arranged inside a given surface such that the packing density is maximized and none of the circles overlap. As an illustrative example, Figure \ref{Packing} shows the optimal packing of 3 equal circles inside a bigger circle. 
Also, Table I shows the radii of non-overlapping small circles which lead to the maximum packing density for a given circular area \cite{gaspar}. Clearly, the radius of each circle decreases as the number of circles increases. In Table I, the total coverage represents the maximum portion of the desired area which can be covered by multiple circles. \textcolor{black}{Note that, in general, the circle packing problem in a bounded area is known to be intractable \cite{gaspar}}. In particular, it is not possible to find  a general packing strategy that is optimal for any arbitrary $M$. Therefore, for each value of $M$ a specific packing strategy needs to be provided. As an example, we derive the optimal packing method for $M=3$. Consider a circular area with radius $R_c$. 
Clearly, the packing density is maximized if the maximum distance between the farthest circles is minimized. For $M=3$, all the circles tangent with each other if their centers are placed on the vertices of a equilateral triangle. Hence, Figure \ref{Packing} corresponds to the optimal placement. Then, from Figure \ref{Packing},
$x = \frac{{{r_u}}}{{\cos ({{30}^o})}}$ and,\\
 ${R_c} = {r_u} + x = {r_u}\left( {1 + \frac{2}{{\sqrt 3 }}} \right) \to {r_u} = \frac{{\sqrt 3 {R_c}}}{{2 + \sqrt 3 }} \approx 0.464{R_c}$.\vspace{0.1cm}
   
 In our model, each circle corresponds to the coverage region of each UAV, and maximizing packing density is related to maximizing the coverage area with non-overlapping smaller circles. Therefore, given the radius of the desired area and the number of symmetric UAVs, we can determine the required coverage radii of UAVs as well as their 3D locations which lead to the maximum coverage. 
Subsequently, based on the required coverage range of each UAV ($r_u$), the minimum transmit power of UAVs can be computed. 
Note that, the UAV's altitude should be adjusted based on the coverage radius and the antenna beamwidth by using $h = \frac{{{r_u}}}{{\tan ({\theta _B}/2)}}$.

Next, as a function of the number of UAVs, we derive an upper bound for the altitude which guarantees the non-overlapping condition between the UAVs' coverage regions.\vspace{-0.1cm}
 \begin{proposition}
 	\normalfont
 	Given $M$ UAVs, and $R_c$, the radius of the desired area, an upper bound for the maximum UAVs' altitude for which the coverage overlap does not occur, is given by: \vspace{-0.1cm} 
 	\begin{equation}
 	h \le \frac{{{q_m}{R_c}}}{{\left( {2 + {q_m}} \right)\tan ({\textstyle{{{\theta _B}} \over 2}})}},
 	\end{equation}
\textcolor{black}{ where $q_m$ is the maximum value of variable $q \in \mathds{R}$ that satisfies the following inequality:\\
 \small $ {{\frac{\pi }{{{{\sin }^{ - 1}}(q/2)}}\left( {\frac{{q\sqrt 3  + \sqrt {4 - {q^2}} }}{q}} \right) + \sqrt {12} (1 - M) \ge 0}}$.}
 \end{proposition}
 \begin{proof}
 The maximum packing density ($D$) that can be achieved for a circle, using $M$ equal smaller circles, is upper bounded by \cite{gaspar}:\vspace{-0.3cm}
 \begin{equation}
   D \le \frac{{Mq_m^2}}{{{{\left( {2 + {q_m}} \right)}^2}}}.\vspace{-0.2cm}
   \end{equation}
   Also, clearly, $D = \frac{{Mr_u^2}}{{R_c^2}}$, and hence, $\frac{{Mr_u^2}}{{R_c^2}} \le \frac{{Mq_m^2}}{{{{\left( {2 + {q_m}} \right)}^2}}}$.\\
   Finally, considering ${r_u} \le \frac{{{q_m}{R_c}}}{{\left( {2 + {q_m}} \right)}}$, and $\tan (\frac{{{\theta _B}}}{2}) = \frac{{{r_u}}}{h}$, we have $h \le \frac{{{q_m}{R_c}}}{{\left( {2 + {q_m}} \right)\tan ({\textstyle{{{\theta _B}} \over 2}})}}$.\vspace{-0.2cm}
 \end{proof}
 Proposition 2, provides a necessary condition on the UAVs' altitude for avoiding an overlapping coverage. \vspace{-0.3cm} 

\begin{table}[!t]
	\normalsize
	\begin{center}
	\caption{\small Covering a circular area with radius $R_c$ using identical UAVs- the \emph{circle packing in a circle}  approach.\vspace{-0.1cm}}
	\label{TableP}
	\resizebox{7.2cm}{!}{
	\begin{tabular}{|c|c|c|}
		\hline
		\textbf{Number of UAVs} & \textbf{Coverage radius of each UAV} & \textbf{Maximum total coverage} \\ \hline \hline
	1	&     $R_c$      &      1     \\ \hline 
	2	&     0.5$R_c$       &       0.5    \\ \hline
	3	&     0.464$R_c$      &      0.646     \\ \hline
	4	&     0.413$R_c$      &        0.686   \\ \hline
	5	&     0.370$R_c$     &        0.685   \\ \hline
	6	&     0.333$R_c$      &          0.666 \\ \hline
	7	&     0.333$R_c$    &         0.778  \\ \hline
	8	&     0.302$R_c$    &         0.733  \\ \hline
	9	&     0.275$R_c$     &         0.689  \\ \hline
	10	&     0.261$R_c$    &         0.687  \\ \hline
	\end{tabular}
}
\end{center}\vspace{-0.6cm}
\end{table}

 \begin{figure}[!t]
 	\begin{center}
 		\vspace{-0.1cm}
 		\includegraphics[width=3.5cm]{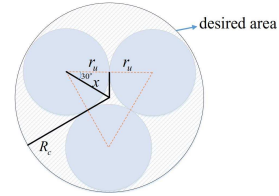}
 		\vspace{-0.1cm}
 		\caption{ \small Packing problem in a circle with 3 circles.\vspace{-0.1cm}}
 		\label{Packing}
 	\end{center}\vspace{-0.68cm}
 \end{figure}
 
\section{Simulation Results and Analysis}\vspace{-0.1cm}
In our simulations, we consider the UAV-based communications over 2\,GHz carrier frequency ($f_c=2\,\rm GHz$)  in an urban environment with $\alpha=0.6$, $\gamma=0.11$, $k_1=10.39$, $k_2=0.05$, $g_1=29.06$, $g_2=0.03$, $\mu_{\rm LoS}=1\,\rm dB$, $\mu_{\rm NLoS}=20\, \rm dB$, and $n=2.5$ \cite{HouraniModeling}. Moreover, using Theorem 1 and (\ref{ru}), the coverage radius of each UAV is computed based on $\epsilon=0.80$, $\beta=5$, and $N=-120$\,dBm. Furthermore, here, the maximum coverage time duration of the UAVs (coverage lifetime) is inversely proportional to the transmit power of UAVs. 

\textcolor{black}{Figure \ref{Pt} shows the total coverage and the coverage lifetime as a function of the number of UAVs ($M$) for $R_c=5000\,\rm m$, and $\theta_B=80^o$. In this figure, the maximum achievable coverage while maximizing the coverage lifetime is shown for different number of UAVs. Clearly, by increasing the number of UAVs, the coverage lifetime increases due to the decrease in the transmit power of each UAV. In Figure \ref{Pt}, for the given area with a radius 5000\,m, a single UAV has the maximum coverage performance. However, in this case, the single UAV has a minimum coverage lifetime. Therefore, depending on the size of the area, coverage and coverage lifetime requirements, an appropriate   number of UAVs needs to be deployed.} 


Figure \ref{h} illustrates the optimal UAVs' altitude versus the number of UAVs. Clearly, the altitude of UAVs should be decreased as the number of UAVs increases. For higher number of UAVs, to avoid overlapping between the coverage regions of the UAVs, the coverage radius of UAVs must be decreased by reducing their height according to $h = \frac{{{r_u}}}{{\tan ({\theta _B}/2)}}$. As shown, by doubling the number of UAVs from 3 to 6, the optimal altitude is reduced from 2000\,m  to 1300\,m. Furthermore, the optimal altitude is lower for higher antenna beamwidths.     

Figure \ref{Nu} shows the minimum required number of UAVs in order satisfy the coverage requirement of the given geographical area. In this figure, the coverage threshold corresponds to the minimum portion of the given area which needs to be covered by the UAVs. This result is based on $P_t=35\,\rm dBm$, $\theta_B=80^o$, and optimal altitudes subject to $h<$\,5000\,\rm m.  
 Interestingly, to satisfy at least 0.7 coverage requirement with a maximum coverage lifetime, either one UAV or more than 6 UAVs are required. In other words, for 1\,$<$$M$$<$\,7, the 0.7 coverage performance cannot be achieved. In general, as the size of the desired area increases, more UAVs are needed to meet the coverage requirement. Clearly, 
 for $R_c$$<$\,5400\,\rm m, using a single UAV can satisfy a 0.6 coverage threshold. However, for a larger target area, more UAVs must be used to reach the coverage threshold. Therefore, the optimal number of UAVs for an efficient system design is significantly dependent on the coverage requirement, and the size of geographical area.\vspace{-0.49cm}  

\begin{figure}[!t]
	\begin{center}
		\vspace{-0.3cm}
		\includegraphics[width=7.3cm]{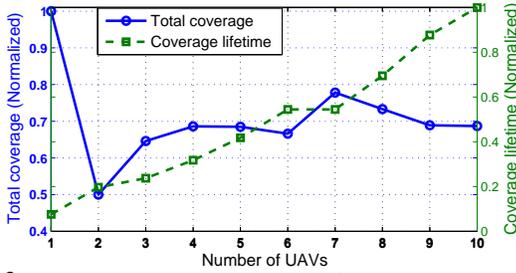}
		\vspace{-0.3cm}
		\caption{\small Total coverage and coverage lifetime versus number of UAVs for $R_c=5000$\,m.}
		\label{Pt}
	\end{center}\vspace{-0.83cm}
\end{figure}

\begin{figure}[!t]
	\begin{center}
		\vspace{-0.2cm}
		\includegraphics[width=7.2cm]{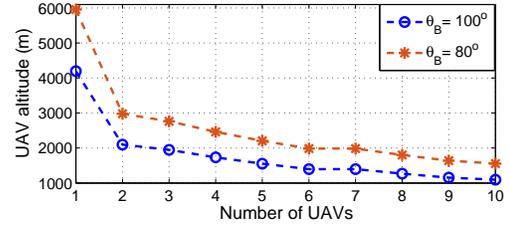}
		\vspace{-0.2cm}
		\caption{ \small UAV's altitude versus number of UAVs.\vspace{-0.1cm}}
		\label{h}
	\end{center}\vspace{-0.42cm}
\end{figure}

\begin{figure}[!t]
	\begin{center}
		\vspace{-0.2cm}
		\includegraphics[width=7.3cm]{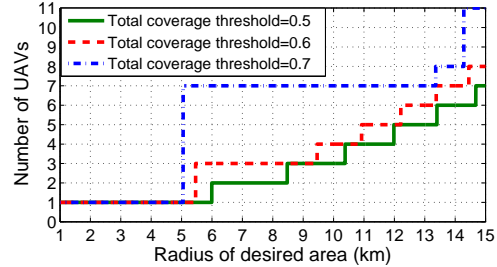}
		\vspace{-0.2cm}
		\caption{ \small Number of required UAVs versus radius of desired area.\vspace{-0.14cm}}
		\label{Nu}
	\end{center}\vspace{-0.70cm}
\end{figure}

\section{Conclusions}\vspace{-0.1cm}
\textcolor{black}{In this paper, we have studied the optimal deployment of multiple UAVs equipped with directional antennas used as aerial base stations. First, the downlink coverage probability was derived based on the probabilistic LoS/NLoS links and considering the shadow fading. Next, given a desired geographical area which needs to be covered by multiple UAVs, an efficient deployment approach was proposed based on the circle packing theory that leads to a maximum coverage while each UAV uses a minimum transmit power. The results have shown that, the optimal altitude and location of the UAVs can be determined based on the number of available UAVs and the gain/beamwidth of the directional antennas.  \vspace{-0.20cm}
}

\def\baselinestretch{1}
\bibliographystyle{IEEEtran}
\vspace{-0.3cm}
\bibliography{references}
\vspace{-0.9cm}
\end{document}